\documentclass[letterpaper, 10pt, twocolumn]{ieeeconf}      

\IEEEoverridecommandlockouts                              
\def\BibTeX{{\rm B\kern-.05em{\sc i\kern-.025em b}\kern-.08em
    T\kern-.1667em\lower.7ex\hbox{E}\kern-.125emX}}
    
\usepackage[left=54pt, right=54pt,bottom=54pt, top=54pt]{geometry}

\usepackage{amsmath,mathrsfs,amsfonts,amssymb,graphicx,epsfig,stmaryrd}
 \usepackage{amsthm}
\usepackage{subcaption}
\usepackage{caption}
\usepackage{color,multirow,rotating}
\usepackage{algorithm,algpseudocode,algorithmicx}
\usepackage{cite,url,framed,bm,balance,dsfont,varwidth}
\usepackage{nicefrac,stmaryrd}

\newtheorem{proposition}{Proposition}
\newtheorem{remark}{Remark}

\usepackage{cleveref}
\usepackage{tikz}
\usetikzlibrary{calc,trees,positioning,arrows,chains,shapes.geometric,decorations.pathreplacing,decorations.pathmorphing,shapes, matrix,shapes.symbols}

\DeclareMathOperator*{\proj}{proj}

\newcommand{\differential}{{\rm{d}}}

\newcommand{\diag}{\mathrm{diag}}

\newcommand{\RNum}[1]{\uppercase\expandafter{\romannumeral #1\relax}}

\allowdisplaybreaks

\title{\LARGE\textbf{
Path Structured Multimarginal Schr\"{o}dinger Bridge for Probabilistic\\Learning of Hardware Resource Usage by Control Software}
}

\author{Georgiy A. Bondar, Robert Gifford, Linh Thi Xuan Phan, Abhishek Halder
\thanks{Georgiy A. Bondar is with the Department of Applied Mathematics, University of California, Santa Cruz, CA 95064, USA,
       {\tt\small{gbondar@ucsc.edu}}.\\
       Robert Gifford and Linh Thi Xuan Phan are with the Department of Computer and Information Science, University of Pennsylvania, Philadelphia, PA 19104, USA,
        {\tt\small{\{rgif,linhphan\}@seas.upenn.edu}}.\\
        Abhishek Halder is with the Department of Aerospace Engineering, Iowa State University, Ames, IA 50011, USA,
        {\tt\small{ahalder@iastate.edu}}.\\
        This work was supported by NSF grants 2112755, 2111688 and  1750158.
}}

\begin{document}

\maketitle

\bstctlcite{IEEEexample:BSTcontrol} 

\begin{abstract}
Solution of the path structured multimarginal Schr\"{o}dinger bridge problem (MSBP) is the most-likely measure-valued trajectory consistent with a sequence of observed probability measures or distributional snapshots. We leverage recent algorithmic advances in solving such structured MSBPs for learning stochastic hardware resource usage by control software. The solution enables predicting the time-varying distribution of hardware resource availability at a desired time with guaranteed linear convergence. We demonstrate the efficacy of our probabilistic learning approach in a model predictive control software execution case study. The method exhibits rapid convergence to an accurate prediction of hardware resource utilization of the controller. The method can be broadly applied to any software to predict cyber-physical context-dependent performance at arbitrary time. 
\end{abstract}

\section{Introduction}\label{sec:introduction}

\noindent Control software in safety-critical cyber-physical systems (CPS) is often designed and verified based on platform models that do not fully capture the complexity of its deployment settings. For example, it is common to assume that the processor always operates at full speed, is dedicated to the control software, and that overheads are negligible. In practice, the hardware resources -- such as last-level shared cache (LLC), memory bandwidth and processor cycles -- often vary with time and depend on the current hardware state, which is a reason why we observe different execution times across different runs of the same control software~\cite{Bernat2002}. This gap can lead to overly costly or unsafe design.

Measurement-based approaches and overhead-aware analysis can  reduce the analysis pessimism or ensure safety~\cite{Mingsong2009}. The recent work~\cite{gifford2021dna} uses fine-grained profiles of the software execution on an actual platform to make dynamic scheduling and resource allocations. Supervisory algorithms that dynamically switch among a bank of controllers -- all provably safe but some computationally more benign (and less performant) than others -- depending on the resource availability also exist~\cite{Zhang2016}. However, the effectiveness of these techniques is contingent on the quality of prediction of the hardware resource availability at a future instance or time horizon of interest.

In this work, we propose to predict the resource usage by control software based on just a very small set of measurements. This approach is attractive as it can reduce measurement efforts while better handling potential variances.

A first-principle predictive model for hardware resource availability based on semiconductor physics of the specific platform is, however, unavailable. Furthermore, resources such as cache and bandwidth are not only time-varying and stochastic, but they are also statistically correlated. This makes it challenging to predict the \emph{joint stochastic variability} of the hardware resource availability in general. The challenge is even more pronounced for control software because the computational burden then also depends on additional context, e.g., reference trajectory that the controller is tracking. 

We note that for safety-critical CPS, predicting the joint stochastic hardware resource state, as opposed to predicting a lumped variable such as worst-case execution time, can open the door for designing a new class of dynamic scheduling algorithms with better performance than what is feasible today while minimizing hardware cost.

This work proposes learning a joint stochastic process for hardware resource availability from control software execution profiles conditioned on CPS contexts (to be made precise in Sec. \ref{subsec:Contextc}, \ref{subsec:HardwareResourceStatexi}). Our proposed method leverages recent advances in \emph{stochastic control} -- specifically in the multimarginal Schr\"{o}dinger bridge (MSBP) -- to allow prediction of time-varying joint statistical distributions of hardware resource availability at any desired time.




\subsubsection*{Contributions} Our specific contributions are as follows.
\begin{itemize}

\item We show how recent algorithmic developments in solving the MSBP, enable probabilistic learning of hardware resources. This advances the state-of-the-art at the intersection of control, learning and real-time systems.

\item The proposed method is statistically nonparametric, and is suitable for high-dimensional joint prediction since it avoids gridding the hardware feature/state space.

\item The proposed formulation provably predicts the most likely distribution given a sequence of distributional snapshots for the hardware resource state. 

\item We explain that the resulting algorithm is an instance of the multimarginal Sinkhorn iteration with path structured cost that is guaranteed to converge to a unique solution, and enjoys linear rate of convergence. Its computational complexity scales linearly w.r.t. dimensions, linearly w.r.t. number of distributional snapshots, and quadratically w.r.t. number of scattered samples.  

\end{itemize}

\subsubsection*{Organization} We introduce the notations and preliminaries in Sec. \ref{secNotationsAndPrelim}. The problem formulation is detailed in Sec. \ref{sec:formulation}. Sec. \ref{sec:OverallAlgorithm} summarizes the overall methodology, which is then followed by a numerical case study in Sec. \ref{sec:numerical}. Concluding remarks are provided in Sec. \ref{sec:conclusions}.


\section{Notations and Preliminaries}\label{secNotationsAndPrelim} 
\noindent We use unboldfaced capital letters to denote matrices and bold capital letters to denote tensors (of order three or more). Unboldfaced (resp. boldfaced) small letters are used to denote scalars (resp. vectors). Capital calligraphic letters are reserved to denote sets.

Square braces are used to denote the components. For instance, $\left[\bm{X}_{i_{1},\hdots,i_{r}}\right]$ denotes the $(i_1,\hdots,i_{r})$th component of the order $r$ tensor $\bm{X}$, where $(i_1,\hdots,i_{r})\in\mathbb{N}^{r}$. We use the $r$ fold tensor product space notation $\left(\mathbb{R}^{d}\right)^{\otimes r} := \underbrace{\mathbb{R}^{d} \otimes \hdots \otimes\mathbb{R}^{d}}_{r\;\text{times}}$. 

For two tensors $\bm{X},\bm{Y}$ of order $r$, we define their \emph{Hilbert-Schmidt inner product} as
\begin{align}
\langle\bm{X},\bm{Y}\rangle := \sum_{i_1,\hdots,i_r}\left[\bm{X}_{i_{1},\hdots,i_{r}}\right]\left[\bm{Y}_{i_{1},\hdots,i_{r}}\right].
\label{HilbertSchmidtInnerProdTensors}
\end{align}

The operators $\exp(\cdot)$ and $\log(\cdot)$ are understood elementwise. We use $\odot$ and $\oslash$ to denote elementwise (Hadamard) multiplication and division, respectively. 

For measures $\mu,\nu$ defined on two Polish spaces, their product measure is denoted by $\mu \otimes \nu$. The \emph{relative entropy} a.k.a. \emph{Kullback-Leibler divergence} $D_{\rm{KL}}(\cdot\Vert\cdot)$ between probability measures $\mu$ and $\nu$ is
\begin{align}
D_{\rm{KL}}(\mu\Vert\nu):=\begin{cases}
\int \log\frac{\differential\mu}{\differential\nu}\differential\mu & \text{if}\quad\mu \ll \nu,\\
+\infty &\text{otherwise,}
\end{cases}
\label{DefKLdivergence}
\end{align}
where $\frac{\differential\mu}{\differential\nu}$ denotes the Radon-Nikodym derivative, and $\mu \ll \nu$ is a shorthand for ``$\mu$ is absolutely continuous w.r.t. $\nu$".

The Hilbert (projective) metric (see e.g., \cite{bushell1973hilbert}) $d_{\rm{H}}\left(\bm{u},\bm{v}\right)$ between two vectors $\bm{u},\bm{v}\in\mathbb{R}^{n}_{>0}$ is
\begin{align}
d_{\rm{H}}\left(\bm{u},\bm{v}\right) = \log\left(\dfrac{\max_{i=1,\hdots,n}u_i/v_i}{\min_{i=1,\hdots,n}u_i/v_i}\right).
\label{HilbertMetric}
\end{align}
We use the term ``control cycle" to mean one pass of a feedback control loop. Due to hardware stochasticity, each control cycle completion takes variable amount of time.


\section{Problem Formulation}\label{sec:formulation}
\subsection{Context $\bm{c}$}\label{subsec:Contextc}
\noindent 
We consider a context vector $\bm{c}$ comprised of separable cyber and physical context vectors
\begin{align}
\bm{c} := \begin{pmatrix}
\bm{c}_{\rm{cyber}}\\
\bm{c}_{\rm{phys}}
\end{pmatrix}.
\label{defContextVec} 
\end{align}
In this work, we consider an instance of \eqref{defContextVec} where 
\begin{align}
\bm{c}_{\rm{cyber}} = \begin{pmatrix}
{\text{allocated last-level cache}}\\
{\text{allocated memory bandwidth}}
\end{pmatrix},
\label{defcCyber}
\end{align}
where both features are allocated in blocks of some size, and
\begin{align}
\bm{c}_{\rm{phys}} = y_{\rm{des}}(x)\in {\rm{GP}}\left([x_{\min},x_{\max}]\right),
\label{defcPhysical}
\end{align}
where ${\rm{GP}}$ denotes a Gaussian process over the domain $[x_{\min},x_{\max}]$. We work with a collection of contexts with cardinality $n_{\rm{context}}$, i.e., a sample of contexts $\{\bm{c}^{i}\}_{i=1}^{n_{\rm{context}}}$. 

\subsection{Hardware Resource State $\bm{\xi}$}\label{subsec:HardwareResourceStatexi}
\noindent 
For concreteness, we define a hardware resource state or feature vector used in our numerical case study (Sec. \ref{sec:numerical}):
\begin{align}
\bm{\xi}:=\begin{pmatrix}
\xi_1\\
\xi_2\\
\xi_3
\end{pmatrix}=
\begin{pmatrix}
\text{instructions retired}\\
\text{LLC requests}\\
\text{LLC misses}
\end{pmatrix}.
\label{defFeatureVec} 
\end{align}
The three elements of $\bm{\xi}$ denote the number of CPU instructions, the number of LLC requests, and the number of LLC misses in the last time unit (10 ms in our profiling), respectively. 

We emphasize that our proposed method is not limited by what specific components comprise $\bm{\xi}$. To highlight this flexibility, we describe the proposed approach for $\bm{\xi}\in\mathbb{R}^{d}$ with suitable interpretations for the specific application. 

For a time interval $[0,t]$ of interest, we think of time-varying $\bm{\xi}$ as a continuous time vector-valued stochastic process over subsets of $\mathbb{R}^{d}$. Suppose that $s\in\mathbb{N}, s\geq 2$ snapshots or observations are made for the stochastic state $\bm{\xi}(\tau)$, $0\leq \tau\leq t$, at (possibly non-equispaced) instances
$$\tau_{1}\equiv 0 < \tau_2 < \hdots < \tau_{s-1} < \tau_{s}\equiv t.$$ 
Consider the snapshot index set $\llbracket s \rrbracket := \{1,\hdots,s\}$.
For a fixed context $\bm{c}$, the snapshot observations comprise a sequence of joint probability measures $\{\mu_{\sigma}\}_{\sigma\in\llbracket s\rrbracket}$ satisfying $\int\differential\mu_{\sigma}(\bm{\xi}(\tau_{\sigma})) = 1$. In other words, 
\begin{align}
\bm{\xi}(\tau_{\sigma})\sim\mu_{\sigma} \quad \forall\sigma\in\llbracket s\rrbracket.
\label{FeatureDistribution}
\end{align}
In our application, the data $\{\mu_{\sigma}\}_{\sigma\in\llbracket s\rrbracket}$ comes from control software execution profiles, i.e., by executing the same control software for the same $\bm{c}$ with all parameters and initial conditions fixed. So the stochasticity in $\bm{\xi}(\tau_{\sigma})$ stems from the dynamic variability in hardware resource availability.  

In particular, for finitely many (say $n$) execution profiles, we consider empirical distributions
\begin{align}
\mu_{\sigma} := \dfrac{1}{n}\sum_{i=1}^{n}\delta(\bm{\xi}-\bm{\xi}^{i}(\tau_{\sigma})),
\label{EmpiricalMeasures}
\end{align}
where $\delta(\bm{\xi}-\bm{\xi}^{i}(\tau_{\sigma}))$ denotes the Dirac delta at sample location $\bm{\xi}^{i}(\tau_{\sigma})$ where $i\in\llbracket n\rrbracket$, $\sigma\in\llbracket s\rrbracket$. At any snapshot index $\sigma\in\llbracket s\rrbracket$, the set $\{\bm{\xi}^{i}(\tau_{\sigma})\}_{i=1}^{n}$ is \emph{scattered data}.

Given the data \eqref{FeatureDistribution}-\eqref{EmpiricalMeasures}, we would like to predict the \emph{most likely} hardware resource state statistics
\begin{align}
\bm{\xi}(\tau) \sim \mu_{\tau} \quad\text{for any}\;\tau\in[0,t].
\label{xisim}
\end{align}
Without the qualifier ``most likely", the problem is overdetermined since there are uncountably many \emph{measure-valued continuous curves} over $[0,t]$ that are consistent with the observed data \eqref{FeatureDistribution}-\eqref{EmpiricalMeasures}.

\begin{figure}[t]
    \centering    \includegraphics[width=0.99\linewidth]{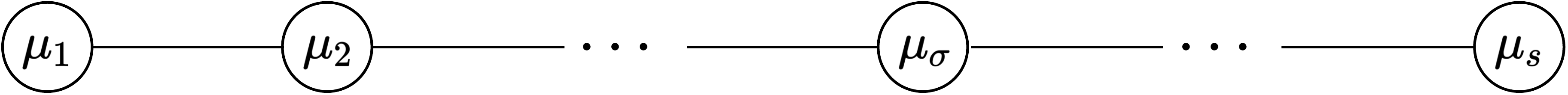}
    \caption{{\small{The path tree for sequentially observed $\{\mu_{\sigma}\}_{\sigma\in\llbracket s\rrbracket}$.}}}
\vspace*{-0.2in}
\label{FigPathTree}
\end{figure}

\subsection{Multimarginal Schr\"{o}dinger Bridge}\label{subsecMSBP}
\noindent
Let $\mathcal{X}_{\sigma} :={\rm{support}}\left(\mu_{\sigma}\right)\subseteq\mathbb{R}^{d}$ $\forall\sigma\in\llbracket s\rrbracket$, and consider the Cartesian product $\mathcal{X}_{1}\times\mathcal{X}_{2}\times \hdots \times \mathcal{X}_{s}=:\boldsymbol{\mathcal{X}}\subseteq\left(\mathbb{R}^{d}\right)^{\otimes s}$. Let $\mathcal{M}\left(\mathcal{X}_{\sigma}\right)$ and $\mathcal{M}\left(\boldsymbol{\mathcal{X}}\right)$ denote the collection (i.e., manifold) of probability measures on $\mathcal{X}_{\sigma}$ and $\boldsymbol{\mathcal{X}}$, respectively. Define a ground cost $\bm{C}:\bm{\mathcal{X}}\mapsto\mathbb{R}_{\geq 0}$.

Following \cite[Sec. 3]{elvander2020multi}, let
{\small{\begin{subequations}
\begin{align}
\differential\bm{\xi}_{-\sigma}&:=\differential\bm{\xi}(\tau_{1})\times\hdots\times\differential\bm{\xi}(\tau_{\sigma-1})\times\differential\bm{\xi}(\tau_{\sigma+1})\times\hdots\times\differential\bm{\xi}(\tau_{s}),\\
\bm{\mathcal{X}}_{-\sigma}&:=\mathcal{X}_{1}\times\hdots\times\mathcal{X}_{\sigma -1}\times\mathcal{X}_{\sigma + 1}\times \hdots \times \mathcal{X}_{s}.    
\end{align}
\label{DefxiX}
\end{subequations}}} 
For $\varepsilon\geq 0$ (not necessarily small), the multimarginal Schr\"{o}dinger bridge problem (MSBP) is the following infinite dimensional convex program:
{\small{\begin{subequations}
\begin{align}
&\underset{\bm{M}\in \mathcal{M}\left(\bm{\mathcal{X}}\right)}{\min} \!\int_{\bm{\mathcal{X}}}\!\big\{\bm{C}\!\left(\bm{\xi}(\tau_{1}),\hdots,\bm{\xi}(\tau_{s})\right)+ \varepsilon\log\bm{M}\!\left(\bm{\xi}(\tau_{1}),\hdots,\bm{\xi}(\tau_{s})\right)\!\!\big\}\nonumber\\
&\qquad\qquad\qquad\bm{M}\!\left(\bm{\xi}(\tau_{1}),\hdots,\bm{\xi}(\tau_{s})\right)\differential\bm{\xi}(\tau_{1}) \hdots\differential\bm{\xi}(\tau_{s})\label{MSBPobj}\\
&\text{subject to}\int_{\bm{\mathcal{X}}_{-\sigma}}\!\!\!\!\!\!\bm{M}\!\left(\bm{\xi}(\tau_{1}), \hdots,\bm{\xi}(\tau_{s})\right)\differential\bm{\xi}_{-\sigma} = \mu_{\sigma}\,\forall\sigma\in\llbracket s\rrbracket.\label{MSBPconstr}
\end{align}
\label{MSBP}
\end{subequations}}}
In particular, $\mathcal{M}(\bm{\mathcal{X}})$ is a convex set. The objective \eqref{MSBPobj} is strictly convex in $\bm{M}$, thanks to the $\varepsilon$-regularized negative entropy term $\int_{\bm{\mathcal{X}}}\varepsilon \bm{M}\log\bm{M}$. The constraints \eqref{MSBPconstr} are linear.

In this work, the measures $\{\mu_{\sigma}\}_{\sigma\in\llbracket s\rrbracket}$ correspond to sequential observation, and we therefore fix the \emph{path structured} (Fig. \ref{FigPathTree}) ground cost
\begin{align}
\bm{C}\!\left(\bm{\xi}(\tau_{1}),\hdots,\bm{\xi}(\tau_{s})\right) = \sum_{\sigma=1}^{s-1}c_{\sigma}\left(\bm{\xi}(\tau_{\sigma}),\bm{\xi}(\tau_{\sigma+1})\right). 
\label{GroundCost}
\end{align}
In particular, we choose the squared Euclidean distance sequential cost between two consecutive snapshot indices, i.e., $c_{\sigma}(\cdot,\cdot):=\|\cdot - \cdot\|_2^2$ $\forall\sigma\in\llbracket s\rrbracket$. MSBPs with more general \emph{tree structured} ground costs have appeared in \cite{haasler2021multimarginal}. 

When the cardinality of the index set $\llbracket s\rrbracket$ equals $2$, then \eqref{MSBP} reduces to the (bi-marginal) Schr\"{o}dinger bridge problem (SBP) \cite{leonard2014survey,chen2021stochastic}. In this case, the solution of \eqref{MSBP} gives the \emph{most likely evolution} between two marginal snapshots $\mu_1,\mu_2$. This can be established via the large deviations \cite{dembo2009large} interpretation \cite[Sec. II]{follmer1988random} of SBP using Sanov's theorem \cite{sanov1958probability}; see also \cite[Sec. 2.1]{pavon2021data}. 

Specifically, let $\mathcal{C}\left([\tau_1,\tau_2],\mathbb{R}^{d}\right)$ denote the collection of continuous functions on the time interval $[\tau_1,\tau_2]$ taking values in $\mathbb{R}^{d}$. Let $\Pi(\mu_1,\mu_2)$ be the collection of all path measures on $\mathcal{C}\left([\tau_1,\tau_2],\mathbb{R}^{d}\right)$ with time $\tau_1$ marginal $\mu_1$, and time $\tau_2$ marginal $\mu_2$. Given a symmetric ground cost (e.g., Euclidean distance) $C:\mathcal{X}_1\times\mathcal{X}_2\mapsto\mathbb{R}_{\geq 0}$, let 
\begin{align}
K(\cdot,\cdot):=\exp\left(-\dfrac{C(\cdot,\cdot)}{\varepsilon}\right),
\label{defkBi}
\end{align} 
and consider the \emph{bimarginal Gibbs kernel} 
\begin{align}
K\left(\bm{\xi}(\tau_1),\bm{\xi}(\tau_2)\right)\mu_1\otimes\mu_2.
\label{DefGibbsKernelBi} 
\end{align}
Then, the bimarginal SBP solves  
\begin{align}
\underset{\pi\in\Pi(\mu_1,\mu_2)}{\min}\varepsilon D_{\rm{KL}}\left(\pi\Vert K\left(\bm{\xi}(\tau_1),\bm{\xi}(\tau_2)\right)\mu_1\otimes\mu_2\right),
\label{BimarginalKL}
\end{align}
i.e., the most likely evolution of the path measure consistent with the observed measure-valued snapshots $\mu_1,\mu_2$.

Under the stated assumptions on the ground cost $c$, the existence of minimizer for \eqref{BimarginalKL} is guaranteed \cite{csiszar1975divergence,borwein1994entropy}. The uniqueness of minimizer follows from strict convexity of the map $\pi \mapsto D_{\rm{KL}}(\pi\Vert\nu)$ for fixed $\nu$.

This relative entropy reformulation, and thereby ``the most likely evolution consistent with observed measures" interpretation, also holds for the MSBP \eqref{MSBP} with $s\geq 2$ snapshots. Specifically, for $\bm{C}:\bm{\mathcal{X}}\mapsto\mathbb{R}_{\geq 0}$ as in \eqref{MSBP}-\eqref{GroundCost}, we generalize \eqref{defkBi} as
\begin{align}
\bm{K}\!\left(\!\bm{\xi}(\tau_1),\hdots,\bm{\xi}(\tau_s)\!\right) :=\exp\!\left(\!\!-\dfrac{\bm{C}\!\left(\bm{\xi}(\tau_1),\hdots,\bm{\xi}(\tau_s)\!\right)}{\varepsilon}\!\!\right),
\end{align}
and define the \emph{multimarginal Gibbs kernel}
\begin{align}
\bm{K}\left(\bm{\xi}(\tau_1),\hdots,\bm{\xi}(\tau_s)\right)\!\mu_1\otimes\hdots\otimes\mu_s.
\label{DefGibbsKernelMulti} 
\end{align}
Problem \eqref{BimarginalKL} then generalizes to
\begin{align}
\underset{\pi\in\Pi(\mu_1,\hdots,\mu_s)}{\min}\varepsilon D_{\rm{KL}}\left(\pi\Vert \bm{K}\left(\bm{\xi}(\tau_1),\hdots,\bm{\xi}(\tau_s)\right)\mu_1\otimes\hdots\otimes\mu_s\right)
\label{MultimarginalKL}
\end{align}
where $\Pi(\mu_1,\hdots,\mu_s)$ denotes the collection of all path measures on $\mathcal{C}\left([\tau_1,\tau_s],\mathbb{R}^{d}\right)$ with time $\tau_{\sigma}$ marginal $\mu_\sigma$ $\forall\sigma\in\llbracket s\rrbracket$. The equivalence between \eqref{MSBP} and \eqref{MultimarginalKL} can be verified by direct computation. Thus solving \eqref{MultimarginalKL}, or equivalently \eqref{MSBP}, yields the most likely evolution of the path measure consistent with the observed measure-valued snapshots $\mu_\sigma$ $\forall\sigma\in\llbracket s\rrbracket$.

We propose to solve the MSBP \eqref{MSBP} for learning the time-varying statistics of the hardware resource state $\bm{\xi}$ as in \eqref{xisim}. We next detail a discrete formulation to numerically solve the same for scattered data $\{\bm{\xi}^{i}(\tau_{\sigma})\}_{i=1}^{n}$ where $n$ is the number of control software execution profiles.

The minimizer of \eqref{MSBP}, $\bm{M}_{\rm{opt}}\left(\bm{\xi}(\tau_{1}), \hdots,\bm{\xi}(\tau_{s})\right)$ can be used to compute the optimal coupling between snapshot index pairs $(\sigma_1,\sigma_2)\in\{\llbracket s\rrbracket^{\otimes 2}\mid\sigma_1 < \sigma_2\}$ as
\begin{align}
\int_{\bm{\mathcal{X}}_{-\sigma_{1},-\sigma_{2}}}\!\!\!\!\!\!\bm{M}_{\rm{opt}}\!\left(\bm{\xi}(\tau_{1}), \hdots,\bm{\xi}(\tau_{s})\right)\differential\bm{\xi}_{-\sigma_{1},-\sigma_{2}}
\label{bimarginalProj}
\end{align}
where
{\small{\begin{subequations}
\begin{align}
\differential\bm{\xi}_{-\sigma_{1},-\sigma_{2}}&:=\prod_{\sigma\in\llbracket s\rrbracket\setminus\{\sigma_1,\sigma_2\}}\differential\bm{\xi}(\tau_{\sigma}),\\
\bm{\mathcal{X}}_{-\sigma_{1},-\sigma_{2}}&:=\prod_{\sigma\in\llbracket s\rrbracket\setminus\{\sigma_1,\sigma_2\}} \mathcal{X}_{\sigma}.    
\end{align}
\label{DefxiXDouble}
\end{subequations}}}
This will be useful for predicting the statistics of $\bm{\xi}(\tau)\sim\mu_{\tau}$ at any (out-of-sample) query time $\tau\in[0,t]$. 

\begin{remark}\label{RemarkMOT}
(\textbf{MSBP and MOT}) When the entropic regularization strength $\varepsilon = 0$, then the MSBP \eqref{MSBP} reduces to the multimarginal optimal transport (MOT) problem \cite{ruschendorf2002n,pass2015multi} that has found widespread applications in barycenter computation \cite{agueh2011barycenters}, fluid dynamics \cite{brenier2008generalized,benamou2019generalized}, team matching problems \cite{carlier2015numerical}, and density functional theory \cite{buttazzo2012optimal,cotar2013density}. Further specializing MOT with the cardinality of $\llbracket s\rrbracket$ equals $2$, yields the (bimarginal) optimal transport \cite{peyre2019computational} problem. 
\end{remark}


\subsection{Discrete Formulation of MSBP}\label{subsecDiscreteFormulation}
\noindent
For finite scattered data $\{\bm{\xi}^{i}(\tau_{\sigma})\}_{i=1}^{n}$ and $\{\mu_{\sigma}\}_{\sigma\in\llbracket s\rrbracket}$ as in \eqref{EmpiricalMeasures}, we set up a discrete version of \eqref{MSBP} as follows. 

With slight abuse of notations, we use the same symbol for the continuum and discrete version of a tensor. The ground cost in discrete formulation is represented by an order $s$ tensor $\bm{C}\in\left(\mathbb{R}^{n}\right)^{\otimes s}_{\geq 0}$, with components
$\left[\bm{C}_{i_{1},\hdots,i_{s}}\right] = \bm{C}\left(\bm{\xi}_{i_1},\hdots,\bm{\xi}_{i_{s}}\right)$. The component $\left[\bm{C}_{i_{1},\hdots,i_{s}}\right]$ encodes the cost of transporting unit mass for a tuple $(i_1,\hdots,i_s)$.

Likewise, consider the discrete mass tensor $\bm{M}\in\left(\mathbb{R}^{n}\right)^{\otimes s}_{\geq 0}$ with components $\left[\bm{M}_{i_{1},\hdots,i_{s}}\right] = \bm{M}\left(\bm{\xi}_{i_1},\hdots,\bm{\xi}_{i_{s}}\right)$. The component $\left[\bm{M}_{i_{1},\hdots,i_{s}}\right]$ denotes the amount of transported mass for a tuple $(i_1,\hdots,i_s)$.

For any $\sigma\in\llbracket s\rrbracket$, the empirical marginals $\bm{\mu}_{\sigma}\in\mathbb{R}^{n}_{\geq 0}$ are supported on the finite sets $\{\bm{\xi}^{i}(\tau_{\sigma})\}_{i=1}^{n}$. We denote the projection of $\bm{M}\in\left(\mathbb{R}^{n}\right)^{\otimes s}_{\geq 0}$ on the $\sigma$th marginal as $\proj_{\sigma}(\bm{M})$. Thus $\proj_{\sigma}:\left(\mathbb{R}^{n}\right)^{\otimes s}_{\geq 0}\mapsto \mathbb{R}^{n}_{\geq 0}$, and is given componentwise as
\begin{align}
\left[{\rm{proj}}_\sigma \!\left(\bm{M}\right)_{j}\!\right] \!\!=\!\! \!\!\sum_{i_1,\hdots,i_{\sigma-1},i_{\sigma + 1},\hdots, i_s}\!\!\!\!\!\!\!\!\bm{M}_{i_1,\hdots,i_{\sigma-1},j,i_{\sigma + 1},\hdots, i_s}.
\label{DefprojsigmaComponent}
\end{align}
Likewise, denote the projection of $\bm{M}\in\left(\mathbb{R}^{n}\right)^{\otimes s}_{\geq 0}$ on the $(\sigma_1,\sigma_2)$th marginal as $\proj_{\sigma_1,\sigma_2}(\bm{M})$, i.e., $\proj_{\sigma_1,\sigma_2}:\left(\mathbb{R}^{n}\right)^{\otimes s}_{\geq 0}\mapsto \mathbb{R}^{n\times n}_{\geq 0}$, and is given componentwise as
\begin{align}
&\left[{\rm{proj}}_{\sigma_{1},\sigma_{2}} \!\left(\bm{M}\right)_{j,\ell}\!\right]\nonumber\\
&=\!\!\!\!\sum_{i_{\sigma}\mid\sigma\in\llbracket s\rrbracket\setminus\{\sigma_1,\sigma_2\}}\!\!\!\!\!\!\!\!\bm{M}_{i_1,\hdots,i_{\sigma_{1}-1},j,i_{\sigma_{1} + 1},\hdots,i_{\sigma_{2}-1},\ell,i_{\sigma_{2} + 1},\hdots,i_s}.
\label{DefprojsigmaComponentDouble}
\end{align}
We note that \eqref{DefprojsigmaComponent} and \eqref{DefprojsigmaComponentDouble} are the discrete versions of the integrals in \eqref{MSBPconstr} and \eqref{bimarginalProj}, respectively.

With the above notations in place, the discrete version of \eqref{MSBP} becomes
\begin{subequations}
\begin{align}
&\underset{\bm{M}\in\left(\mathbb{R}^{n}\right)^{\otimes s}_{\geq 0}}{\min}~\langle\bm{C}+\varepsilon\log\bm{M},\bm{M}\rangle\label{DiscreteMSBPobj}\\
&\text{subject to}~~{\rm{proj}}_{\sigma}\left(\bm{M}\right) = \bm{\mu}_{\sigma}\quad\forall\sigma\in\llbracket s\rrbracket.\label{DiscereteMSBPconstr}
\end{align}
\label{DiscreteMSBP}
\end{subequations}
The primal formulation \eqref{DiscreteMSBP} has $n^{s}$ decision variables, and is computationally intractable. Recall that even for the bimarginal ($s=2$) case, a standard approach \cite{cuturi2013sinkhorn} is to use Lagrange duality to notice that the optimal mass matrix $M_{\rm{opt}}$ is a diagonal scaling of $K:=\exp(-C/\varepsilon)\in\mathbb{R}^{n\times n}_{> 0}$, i.e., $M_{\rm{opt}}=\diag(\bm{u}_1) K\diag(\bm{u}_2)$ where $\bm{u}_1 := \exp(\bm{\lambda}_1/\varepsilon)$, $\bm{u}_2 := \exp(\bm{\lambda}_2/\varepsilon)$, and $\bm{\lambda}_1,\bm{\lambda}_2\in\mathbb{R}^{n}$ are the Lagrange multipliers associated with respective bimarginal constraints $\proj_{1}(M)=\bm{\mu}_1$, $\proj_{2}(M)=\bm{\mu}_2$. The unknowns $\bm{u}_1,\bm{u}_2$ can be obtained by performing the Sinkhorn iterations
\begin{subequations}
\begin{align}
\bm{u}_1 &\leftarrow \bm{\mu}_1\oslash \left(K\bm{u}_2\right),\\
\bm{u}_2 &\leftarrow \bm{\mu}_2\oslash \left(K^{\top}\bm{u}_1\right),
\end{align}
\label{BimarginalSink}
\end{subequations}
with guaranteed linear convergence \cite{franklin1989scaling} wherein the computational cost is governed by two matrix-vector multiplications.

The duality result holds for the multimarginal ($s\geq 2$) case. Specifically, the optimal mass tensor in \eqref{DiscreteMSBP} admits a structure $\bm{M}_{\rm{opt}} = \bm{K}\odot\bm{U}$ where $\bm{K}:=\exp(-\bm{C}/\varepsilon)\in\left(\mathbb{R}^{n}\right)^{\otimes s}_{> 0}$, $\bm{U}:=\otimes_{\sigma=1}^{s}\bm{u}_{\sigma}\in\left(\mathbb{R}^{n}\right)^{\otimes s}_{> 0}$, $\bm{u}_{\sigma}:=\exp(\bm{\lambda}_{\sigma}/\varepsilon)$, and $\bm{\lambda}_{\sigma}\in\mathbb{R}^{n}$ are the Lagrange multipliers associated with the respective multimarginal constraints \eqref{DiscereteMSBPconstr}. The unknowns $\bm{u}_{\sigma}$ can, in principle, be obtained from the multimarginal Sinkhorn iterations \cite{benamou2015iterative}
\begin{align}
\bm{u}_{\sigma} \leftarrow \bm{u}_{\sigma} \odot \bm{\mu}_{\sigma}\oslash{\rm{proj}}_{\sigma}\left(\bm{K}\odot\bm{U}\right) ~\forall\sigma\in\llbracket s\rrbracket,
\label{MultimarginalSink}    
\end{align}
which generalize \eqref{BimarginalSink}. However, computing ${\rm{proj}}_{\sigma}\left(\bm{K}\odot\bm{U}\right)$ requires $\mathcal{O}\left(n^{s}\right)$ operations. Before describing how to avoid this exponential complexity (Sec. \ref{subsec:multiSink}), we point out the convergence guarantees for \eqref{MultimarginalSink}.

\subsection{Convergence for Multimarginal Sinkhorn Iterations}\label{subsec:ConvergenceMultiSink}
\noindent
The iterations \eqref{MultimarginalSink} can either be derived as alternating Bregman projections \cite{benamou2015iterative} or via block coordinate dual ascent \cite{elvander2020multi}. Following either viewpoints leads to guaranteed linear convergence of \eqref{MultimarginalSink}; see \cite{bauschke2000dykstras}, \cite[Thm. 3.5]{haasler2021multimarginal}.

More recent works have also established \cite{marino2020optimal} guaranteed convergence for the continuous formulation \eqref{MSBP} with linear rate of convergence \cite{carlier2022linear}.

\subsection{Multimarginal Sinkhorn Iterations for Path Structured $\bm{C}$}\label{subsec:multiSink}
\noindent
We circumvent the exponential complexity in computing ${\rm{proj}}_{\sigma}\left(\bm{K}\odot\bm{U}\right)$ in \eqref{MultimarginalSink} by leveraging the path structured ground cost \eqref{GroundCost}. This is enabled by a key result from \cite{elvander2020multi}, rephrased, and reproved below in slightly generalized form. 
\begin{proposition}\label{PropMultiSinkPathCost}(\cite[Prop. 2]{elvander2020multi})
Consider the discrete ground cost tensor $\bm{C}$ in \eqref{DiscreteMSBP} induced by a path structured cost \eqref{GroundCost} so that $\left[\bm{C}_{i_{1},\hdots,i_{s}}\right] = \sum_{\sigma=1}^{s-1}\left[C_{i_{\sigma},i_{\sigma+1}}^{\sigma\rightarrow\sigma+1}\right]$ where the matrix $C^{\sigma\rightarrow\sigma+1}\in\mathbb{R}^{n\times n}_{\geq 0}$ encodes the cost of transporting unit mass between each source-destination pair from the source set $\{\bm{\xi}^{i}(\tau_{\sigma})\}_{i=1}^{n}$ to the destination set $\{\bm{\xi}^{i}(\tau_{\sigma+1})\}_{i=1}^{n}$.

Let $K^{\sigma\rightarrow\sigma+1}:=\exp(-C^{\sigma\rightarrow\sigma+1}/\varepsilon)\in\mathbb{R}^{n\times n}_{\geq 0}$, $\bm{K}:=\exp(-\bm{C}/\varepsilon)\in\left(\mathbb{R}^{n}\right)^{\otimes s}_{> 0}$, $\bm{U}:=\otimes_{\sigma=1}^{s}\bm{u}_{\sigma}\in\left(\mathbb{R}^{n}\right)^{\otimes s}_{> 0}$.

Then \eqref{DefprojsigmaComponent} and \eqref{DefprojsigmaComponentDouble} can be expressed as
\begin{align}
&{\rm{proj}}_{\sigma}\!\left(\bm{K}\odot\bm{U}\right)\!=\!\!\left(\!\bm{u}_1^{\!\top} K^{1\to2}\!\prod_{j=2}^{\sigma-1}\!\diag(\bm{u}_j)K^{j\to j+1}\!\!\right)^{\!\!\top}\!\!\!\!\odot\bm{u}_{\sigma}\odot \nonumber\\
&\left(\!\!\left(\prod_{j=\sigma+1}^{s-1}K^{j-1\to j}\diag(\bm{u}_j)\!\right)\!\!K^{s-1\to s}\bm{u}_s\!\right)\,\forall\sigma\in\llbracket s\rrbracket,\label{ProjSimplified}
\end{align}
and
\begin{align}
&{\rm{proj}}_{\sigma_1,\sigma_2}\!\left(\bm{K}\odot\bm{U}\right)=\diag\!\left(\!\!\bm{u}_1^\top K^{1\to2}\prod_{j=2}^{\sigma_1-1}\diag(\bm{u}_j)K^{j\to j+1}\!\!\right)\nonumber\\
&\qquad\qquad\qquad\qquad~~~\diag(\bm{u}_{\sigma_1})\!\!\prod_{j=\sigma_1+1}^{\sigma_2}\!\!\left(K^{j-1\to j}\diag(\bm{u}_j)\right)\nonumber\\
&\qquad\qquad\qquad~\diag\!\left(\!\!\left(\prod_{j=\sigma_2+1}^{s-1}K^{j-1\to j}\diag(\bm{u}_j)\!\!\right)\!K^{s-1\to s}\bm{u}_s\!\!\right)\nonumber\\
&\qquad\qquad\qquad\qquad\forall(\sigma_1,\sigma_2)\in\{\llbracket s\rrbracket^{\otimes 2}\mid\sigma_1 < \sigma_2\}.
\label{Proj2Simplified}
\end{align}
\end{proposition}
\begin{proof}
The proof strategy is to write the Hilbert-Schmidt inner product $\langle\bm{K},\bm{U}\rangle$ in two different ways.

First, recall that $\bm{K}:=\exp(-\bm{C}/\varepsilon)\in\left(\mathbb{R}^{n}\right)^{\otimes s}_{> 0}$ and $\bm{U}:=\otimes_{\sigma=1}^{s}\bm{u}_{\sigma}\in\left(\mathbb{R}^{n}\right)^{\otimes s}_{> 0}$. So following \eqref{HilbertSchmidtInnerProdTensors}, we have
\begin{align*}
    \langle\bm{K},\bm{U}\rangle &= \sum_{i_1,\dots,i_s}\left(\prod_{j=2}^{s}\left[K^{j-1\to j}_{i_{j-1},i_j}\right]\right)\prod_{j=1}^s(\bm{u}_j)_{i_j} \\
    &= \sum_{i_1,\dots,i_s}(\bm{u}_1)_{i_1}\prod_{j=2}^{s}\left[K^{j-1\to j}\diag(\bm{u}_j)\right]_{i_{j-1},i_j} \\
    &= \bm{u}_1^\top \left(\prod_{j=2}^{s-1}K^{j-1\to j}\diag(\bm{u}_j)\right)K^{s-1\to s}\bm{u}_s,
\end{align*}
and \eqref{ProjSimplified} follows from \cite[Lemma 1 in Appendix 1]{elvander2020multi}. 

Next, notice that we can alternatively write
\begin{align*}
    \langle\bm{K},\bm{U}\rangle = \:&\bm{u}_1^\top \left(\prod_{j=2}^{\sigma_1-1}K^{j-1\to j}\diag(\bm{u}_j)\right)K^{\sigma_1-1\to \sigma_1}\\
    &\left(\prod_{j=\sigma_1}^{\sigma_2-1}\diag(\bm{u}_j)K^{j\to j+1}\right)\diag(\bm{u}_{\sigma_2})\\
    &\left(\prod_{j=\sigma_2+1}^{s-1}K^{j-1\to j}\diag(\bm{u}_{j})\right)K^{s-1\to s}\bm{u}_s.
\end{align*}
Then \eqref{Proj2Simplified} follows from \cite[Lemma 2 in Appendix 1]{elvander2020multi}.
\end{proof}
\begin{remark}
Unlike \cite[Prop. 2]{elvander2020multi}, our data $\{\bm{\xi}^{i}(\tau_{\sigma})\}_{i=1}^{n}$ are scattered, i.e., not on a fixed grid, hence the need for superscripts $\sigma\rightarrow\sigma + 1$ for the time-varying matrices in our Prop. \ref{PropMultiSinkPathCost}. In contrast, the corresponding matrices in \cite[Prop. 2]{elvander2020multi} are independent of $\sigma$.
\end{remark}
\begin{remark}
We note that substituting \eqref{ProjSimplified} into \eqref{MultimarginalSink} cancels the (elements of) positive vectors $\bm{u}_{\sigma}$ $\forall\sigma\in\llbracket s\rrbracket$ from the corresponding numerators and denominators. This further simplifies our multimarginal Sinkhorn recursions to
\begin{align}
&\bm{u}_{\sigma}\leftarrow \bm{\mu}_{\sigma}\oslash\!\!\left(\!\!\!\left(\!\bm{u}_1^{\!\top} K^{1\to2}\!\prod_{j=2}^{\sigma-1}\!\diag(\bm{u}_j)K^{j\to j+1}\!\!\right)^{\!\!\top}\right.\nonumber\\
&\left.\odot\left(\!\!\left(\prod_{j=\sigma+1}^{s-1}K^{j-1\to j}\diag(\bm{u}_j)\!\right)\!\!K^{s-1\to s}\bm{u}_s\!\!\right)\!\!\right)\:\forall\sigma\in\llbracket s\rrbracket.
\label{MultimargSinkFinal}    
\end{align}
\end{remark}
\begin{remark}(\textbf{From exponential to linear complexity in $s$})
We note that \eqref{MultimargSinkFinal} involves $s-1$ matrix-vector multiplications each of which has $\mathcal{O}(n^2)$ complexity. So the computational complexity for \eqref {MultimargSinkFinal} becomes $\mathcal{O}\left((s-1)n^2\right)$ which is linear in $s$, i.e., a significant reduction from earlier $\mathcal{O}\left(n^{s}\right)$ complexity mentioned at the end of Sec. \ref{subsecDiscreteFormulation}.
\end{remark}
\begin{remark}(\textbf{Linear complexity in $d$})
The dimension $d$ of the vector $\bm{\xi}$ only affects the construction of the time-varying Euclidean distance matrices $C^{\sigma\rightarrow\sigma+1}$ $\forall\sigma\in\llbracket s-1\rrbracket$ in Prop. \ref{PropMultiSinkPathCost}, which has total complexity $\mathcal{O}(sd)$. Once constructed, the recursions \eqref {MultimargSinkFinal} are independent of $d$.
\end{remark}
We next outline how the solution tensor $\bm{M}_{\rm{opt}}=\bm{K}\odot\bm{U}$ obtained from the converged Sinkhorn iterations can be used together with \eqref{Proj2Simplified}, to make stochastic predictions of the most likely hardware resource state in the form \eqref{xisim}.

\subsection{Predicting Most Likely Distribution }\label{subsec:Prediction}
\noindent
For the ground cost \eqref{GroundCost} resulting from sequential information structure (Fig. \ref{FigPathTree}), we utilize \eqref{Proj2Simplified} to decompose $\bm{M}_{\rm{opt}} = \bm{K}\odot\bm{U}$ of \eqref{DiscreteMSBP} into bimarginal transport plans
\begin{align}
    M^{\sigma_1\to\sigma_2}:={\rm{proj}}_{\sigma_1,\sigma_2}(\bm{M}_{\rm{opt}})={\rm{proj}}_{\sigma_1,\sigma_2}(\bm{K}\odot\bm{U}).
\label{BimarginalPlanM}
\end{align}
Further, when $\bm{C}$ is squared Euclidean, as we consider here, the maximum likelihood estimate for $\mu_{\tau}$ in \eqref{xisim} for a query point $\tau\in[0,t]$, is (see \cite[Sec. 2.2]{elvander2020multi})
\begin{align}
    \!\hat{\mu}_\tau \!:=\!\sum_{i=1}^n\!\sum_{j=1}^n \!\left[\!M^{\sigma\to\sigma+1}_{i,j}\!\right]\!\!\delta(\bm{\xi}-\hat{\bm{\xi}}(\tau,\bm{\xi}^{i}(\tau_{\sigma}),\bm{\xi}^{j}(\tau_{\sigma+1})))
\label{MuInterpolation}
\end{align}
where $\sigma\in\llbracket s\rrbracket$ such that $\tau\in[\tau_{\sigma},\tau_{\sigma+1}]$, and
\begin{subequations}
\begin{align}
&\!\hat{\bm{\xi}}(\tau,\bm{\xi}^{i}(\tau_{\sigma}),\bm{\xi}^{j}(\tau_{\sigma+1}))\!\!:= \!(1-\lambda)\bm{\xi}^{i}(\tau_{\sigma}) \!+\! \lambda\bm{\xi}^{j}(\tau_{\sigma+1}),\\
&\lambda:=\dfrac{\tau-\tau_{\sigma}}{\tau_{\sigma+1}-\tau_{\sigma}}\in[0,1].
\end{align}
\label{defXiInterpolation}
\end{subequations}


\section{Overall Algorithm}\label{sec:OverallAlgorithm}
\noindent
The methodology proposed in Sec. \ref{sec:formulation} is comprised of the following three overall steps.\\
\noindent\textbf{Step 1.} Given a collection of contexts (Sec. \ref{subsec:Contextc}) $\{\bm{c}^{i}\}_{i=1}^{n_{\rm{context}}}$, execute the control software over $[0,t]$ to generate hardware resource state sample snapshots (Sec. \ref{subsec:HardwareResourceStatexi}) $\{\bm{\xi}^{i}(\tau_{\sigma})\}_{i=1}^{n}$, and thereby empirical $\mu_{\sigma}$ as in \eqref{EmpiricalMeasures} for all $\sigma\in\llbracket s\rrbracket$, conditional on each of the $n_{\rm{context}}$ context samples.\\
\noindent\textbf{Step 2.} Using data from \textbf{Step 1}, construct Euclidean distance matrices $C^{\sigma\rightarrow \sigma + 1}$ from the source set $\{\bm{\xi}^{i}(\tau_{\sigma})\}_{i=1}^{n}$ to the destination set $\{\bm{\xi}^{i}(\tau_{\sigma+1})\}_{i=1}^{n}$ $\forall \sigma\in\llbracket s-1\rrbracket$. Perform recursions \eqref{MultimargSinkFinal} until convergence (error within desired tolerance).
\\
\noindent\textbf{Step 3.} Given a query context $\bm{c}$ and time $\tau\in[0,t]$, return most likely distribution $\hat{\mu}_{\tau}$ using \eqref{MuInterpolation}.

\begin{remark}\label{ReamrkSteps}
For the three steps mentioned above, \textbf{Step 1} is data generation, \textbf{Step 2} is probabilistic learning using data from \textbf{Step 1}, and \textbf{Step 3} is prediction using the learnt model.
\end{remark}


\begin{figure}[t]
    \centering    \includegraphics[width=0.8\linewidth]{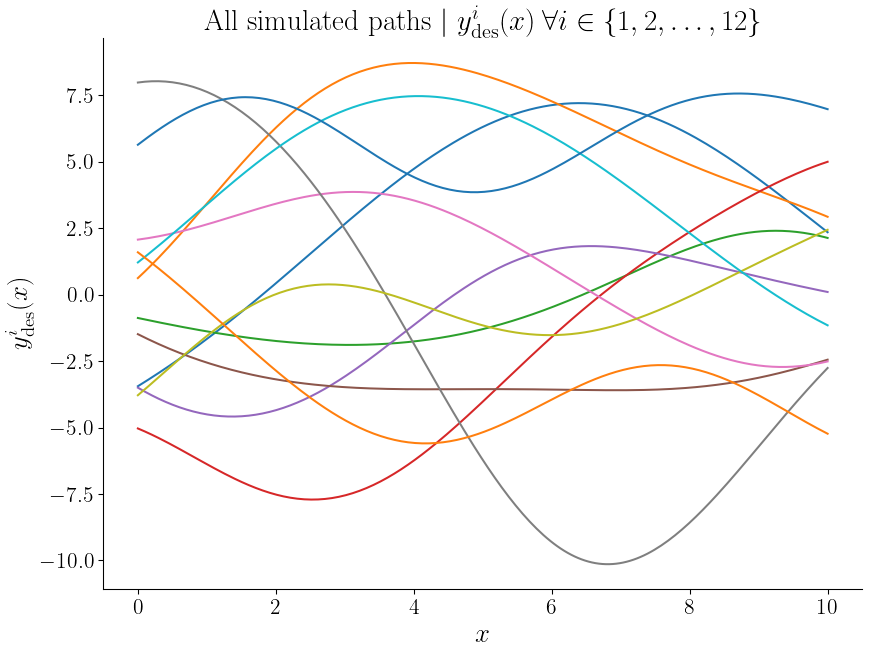}
    \caption{{\small{All 12 paths used in profiling the NMPC software (Sec. \ref{sec:numerical}), generated by GP sampling via Scikit-learn \cite{scikit-learn} over the  domain $[0,10]$ using mean zero and variance 10.}}}
\vspace*{-0.2in}
\label{FigSimPaths}
\end{figure}

\section{Numerical Case Study}\label{sec:numerical}
\noindent
In this Section, we illustrate the application of the proposed method for a vehicle path tracking control software. All along, we provide details for the three steps in Sec. \ref{sec:OverallAlgorithm}.

\noindent\textbf{Control Software.}
We wrote custom software\footnote{Git repo: {\tiny{\url{https://github.com/abhishekhalder/CPS-Frontier-Task3-Collaboration}}}} in C language implementing path following nonlinear model predictive controller (NMPC) for a kinematic bicycle model (KBM) \cite{kong2015kinematic,haddad2022density} of a vehicle with four states $(x,y,v,\psi)$ and two control inputs $(a_c,\delta)$, given by
$\dot{x}=v \cos (\psi+\beta), \dot{y}=v \sin (\psi+\beta), \dot{v}=a_c, \dot{\psi}=\frac{v}{\ell_{\text {rear}}} \sin \beta$,
where the sideslip angle
$\beta=\arctan \left(\frac{\ell_{\text {rear }}}{\ell_{\text {front }}+\ell_{\text {rear }}} \tan \delta\right)$.
The $4\times 1$ state vector comprises of the inertial position $(x, y)$ for the vehicle’s center-of-mass, its speed $v$, and the vehicle’s inertial heading angle $\psi$. The $2\times 1$ control vector comprises of
the acceleration $a_c$, and the front steering wheel angle $\delta$.

\begin{figure*}[t]
    \centering
\includegraphics[width=0.6\linewidth]{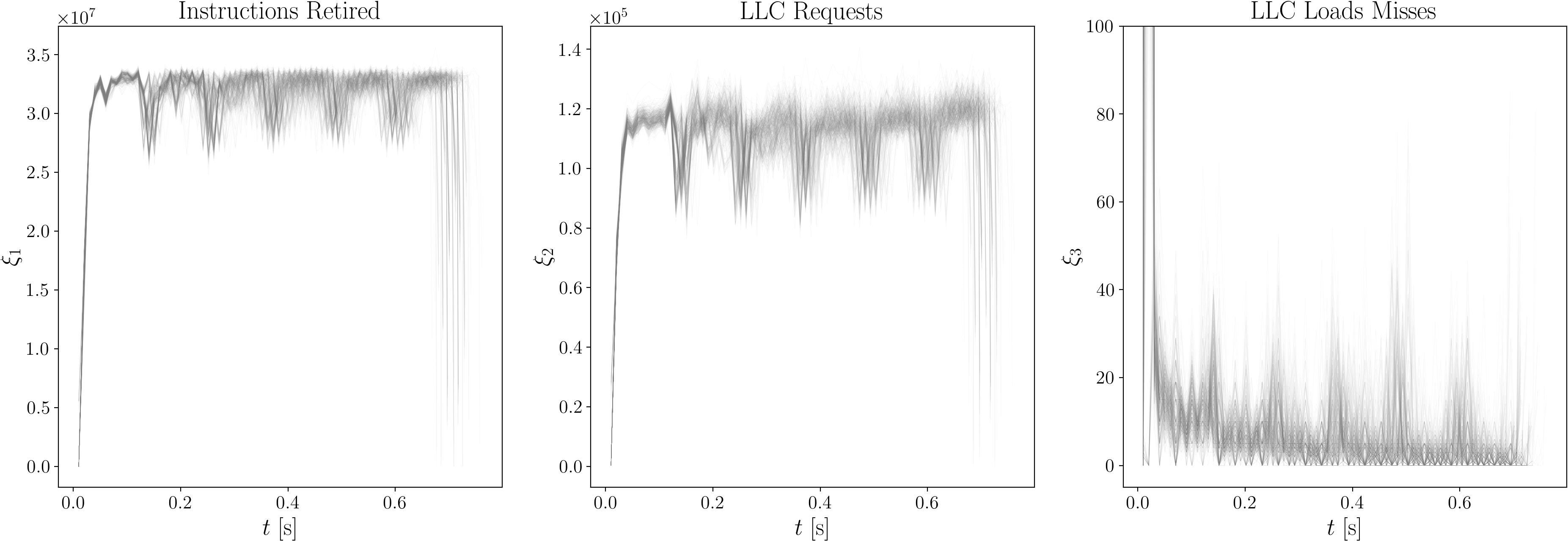}
    \caption{{\small{Components of the measured feature vector $\bm{\xi}$ in \eqref{defFeatureVec} for all of the five control cycles for 500 executions of the NMPC software, where $\bm{c}=[15,15,y^{1}_{\rm{des}}(x)]$.}}}
\vspace*{-0.2in}
\label{FigAllMeasuredStatesForOneContext}
\end{figure*}

The parameters $\ell_{\text {front }},\ell_{\text {rear }}$ are the distances of the vehicle’s
center-of-mass to the front and rear axles, respectively.

The NMPC was designed to track a desired path given as a sequence of $N=200$ waypoint tuples $\big\{(x_d^{(i)},y_d^{(i)},v_d^{(i)})\big\}_{i=1}^{N}$, i.e., a sequence of desired positions and speeds (desired speed profile was numerically estimated from the desired waypoint profiles). At every control step (at most every 100 ms), using the IPOPT nonlinear program solver \cite{IPOPT}, the NMPC solved
an optimization problem to minimize the sum of the crosstrack, $\psi$, and $v$ errors, along with the magnitude and rapidity in change of the control inputs, over a period of time from 0 to the time horizon $H_p=4$, subject to control magnitude and slew rate constraints.
For formulation details and control performance achieved, see \cite{GitDocKBMControllers}. For implementation and parameter values, we refer the readers to the Git repository in the footnote.


While closing the loop with KBM  with computed control values requires minimal computational overhead, the NMPC is computationally demanding. In the case where multiple vehicle controllers are available, it is of practical interest to predict the hardware resource usage for the NMPC for one to several control cycles, conditional on the CPS context $\bm{c}$ (Sec. \ref{subsec:Contextc}) at a given time. For this we `profile' the NMPC, meaning we run the software many times for different values of $\bm{c}$ as in \eqref{defContextVec}, measuring time evolution of the hardware resource state $\bm{\xi}$ as in \eqref{defFeatureVec}. We use these profiles to generate marginals $\mu_{\sigma}$ as in \eqref{EmpiricalMeasures}, thus completing \textbf{Step 1} in Sec. \ref{sec:OverallAlgorithm}.

We next provide details on generating control software execution profiles for our specific case study.\\ 
\noindent\textbf{Generating Execution Profiles.}  To gather the execution profiles for our NMPC control software, we used an Ubuntu 16.04.7 Linux machine with an Intel Xeon E5-2683 v4 CPU.
We leveraged Intel’s Cache Allocation Technology (CAT) \cite{intel-2015-cat} and Memguard \cite{Yun16-memguard-journal} to control allocation of LLC partitions and memory bandwidth available to the control software, respectively. Both LLC partitions and memory bandwidth were allocated in blocks of 2MB.

Utilizing these resource partitioning mechanisms, we ran our application on an isolated CPU and used the Linux perf tool \cite{perf}, version 4.9.3, to sample $\bm{\xi}$ every 10 ms.

For each run of our application, we set the cache and memory bandwidth to a static allocation and pass as input a path for the NMPC to follow, represented as an array of desired $(x,y)$ coordinates.
We then execute the control software for $n_c:=5$ uninterrupted ``control cycles", wherein the NMPC gets the KBM state, makes a control decision, and updates the KBM state.

We profile over 12 unique desired paths to track, denoted $\{y_{\rm{des}}^{i}(x)\}_{i=1}^{12}$, and 5 unique vectors of $\{\bm{c}_{\rm{cyber}}^{i}\}_{i=1}^{5}$, comprising $n_{\rm{context}} = 12 \times 5 = 60$ samples for $\bm{c}$. Conditional on each of these $60$ context samples $\{\bm{c}^{i}\}_{i=1}^{n_{\rm{context}}=60}$, we run the software for 500 profiles for each unique $\bm{c}$ for a total of 30,000 profiles.

The sample paths $\{y_{\rm{des}}^{i}(x)\}_{i=1}^{12}$ in \eqref{defcPhysical} were all generated for $x\in[0,10]$ using a GP with mean zero and variance 10 \cite{scikit-learn}, and are shown in Fig. \ref{FigSimPaths}.

Our vectorial samples $\{\bm{c}_{\rm{cyber}}^{i}\}_{i=1}^{5}$ in \eqref{defcCyber} were $[1,1]^\top$, $[5,5]^\top$, $[10,10]^\top$, $[15,15]^\top$, and $[20,20]^\top$, where each entry represents the number of cache/memory bandwidth partitions from $1$ to $20$. These values were selected to broadly cover the range of possible hardware contexts.

\begin{figure}[t]
    \centering    \includegraphics[width=0.7\linewidth]{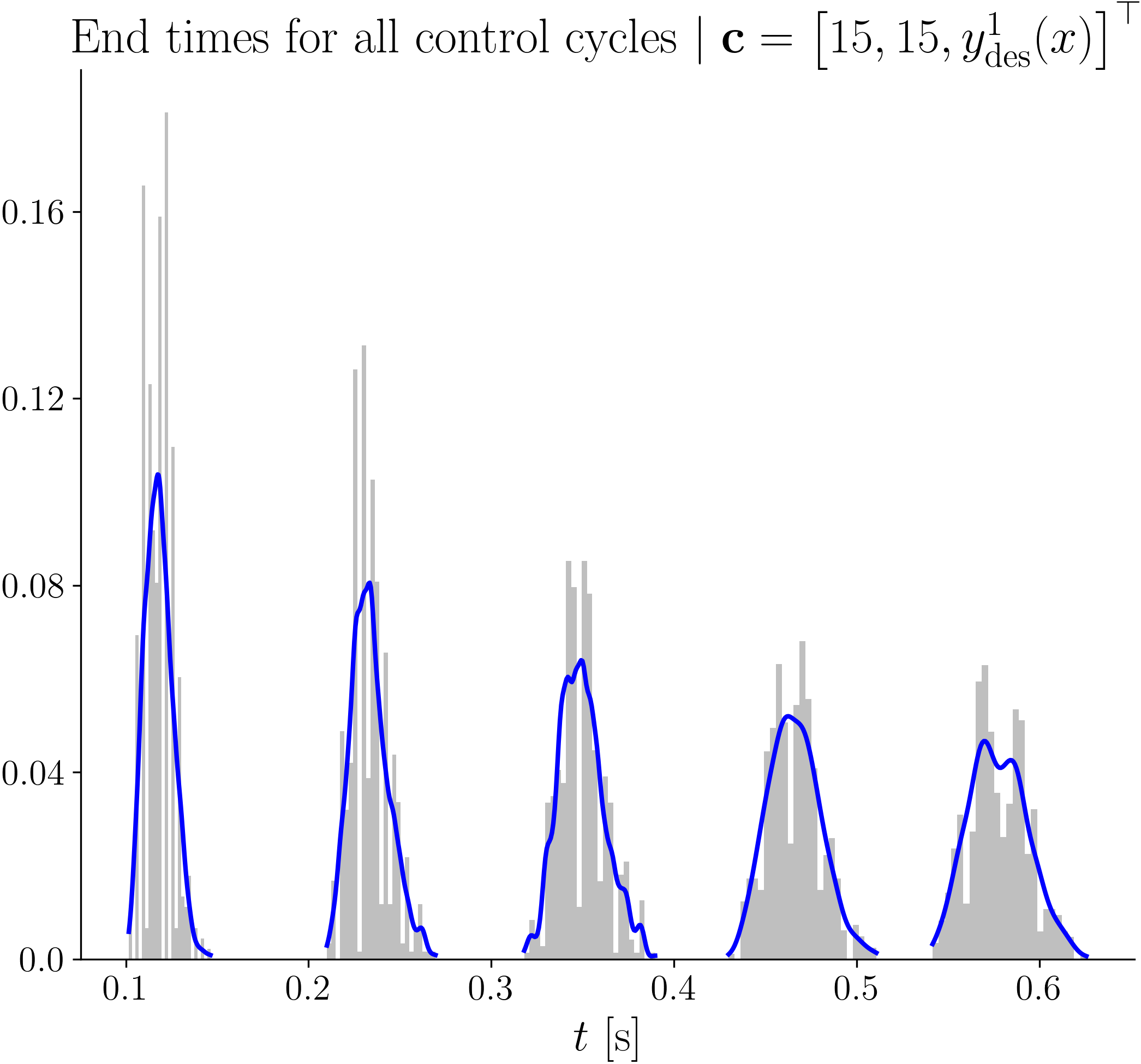}
    \caption{{\small{Normalized histograms (\emph{gray filled}) and kernel density estimates (KDEs) (\emph{solid line}) for the end times of all of the five control cycles for 500 executions of the NMPC software conditioned on a fixed CPS context $\bm{c}$ shown above. The KDEs used Gaussian kernel with bandwidths computed via cross validation \cite{bowman1984alternative,hall1992smoothed}.}}}
\vspace*{-0.2in}
\label{FigControlCycleEndTimesUnnormalizedHistogram}
\end{figure}
\begin{table}[h]
\centering
\begin{tabular}{| c | c | c |} 
 \hline
 ~Control cycle~ & ~Mean~ & ~Standard deviation~\\
 \hline\hline
 \#1 & 0.1181 & 0.0076 \\
 \hline
 \#2 & 0.2336 & 0.0106 \\
 \hline
\#3 & 0.3495 & 0.0127 \\
\hline
\#4 & 0.4660 & 0.0143 \\
\hline
\#5 & 0.5775 & 0.0159 \\
\hline
\end{tabular}
\caption{{\small{The means and standard deviations of the end times for the $n_{c}=5$ control cycles data shown in Fig. \ref{FigControlCycleEndTimesUnnormalizedHistogram}.}}}
\label{tab:ControlCycleEndTimeMeanStdDev}
\vspace*{-0.1in}
\end{table}

\begin{figure}[t]
    \centering    \includegraphics[width=0.9\linewidth]{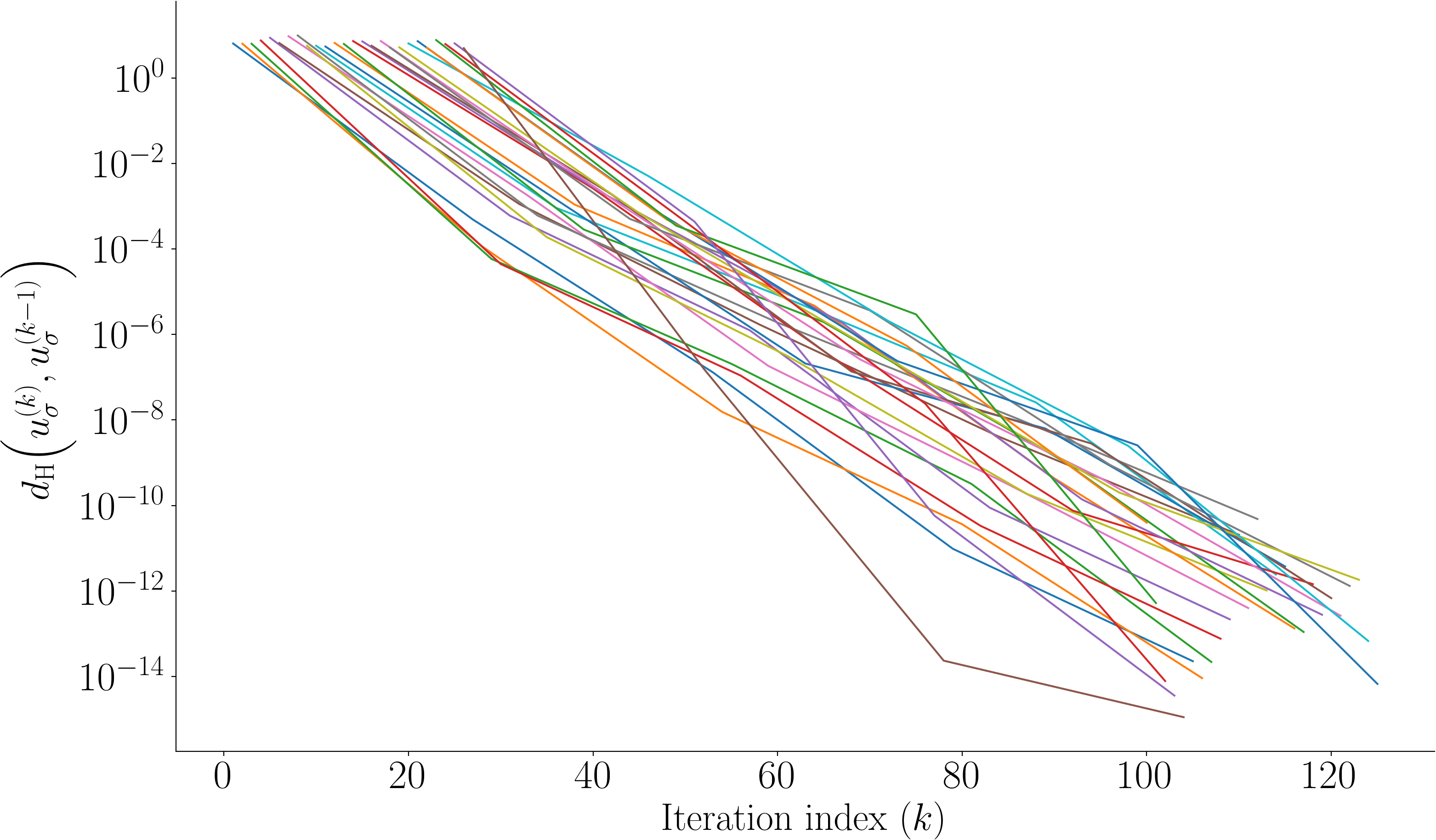}
    \caption{{\small{Linear convergence of Sinkhorn iterations \eqref{MultimargSinkFinal} for $s_{\rm{int}}=4$ w.r.t. the Hilbert's projective metric $d_{\rm{H}}$ in \eqref{HilbertMetric} between $u_{\sigma\in\llbracket s\rrbracket}$ at iteration indices $k$ and $k-1$.}}}
\vspace*{-0.22in}
\label{FigConvergenceHilbert}
\end{figure}
\begin{figure}[h]
    \centering    \includegraphics[width=0.9\linewidth]{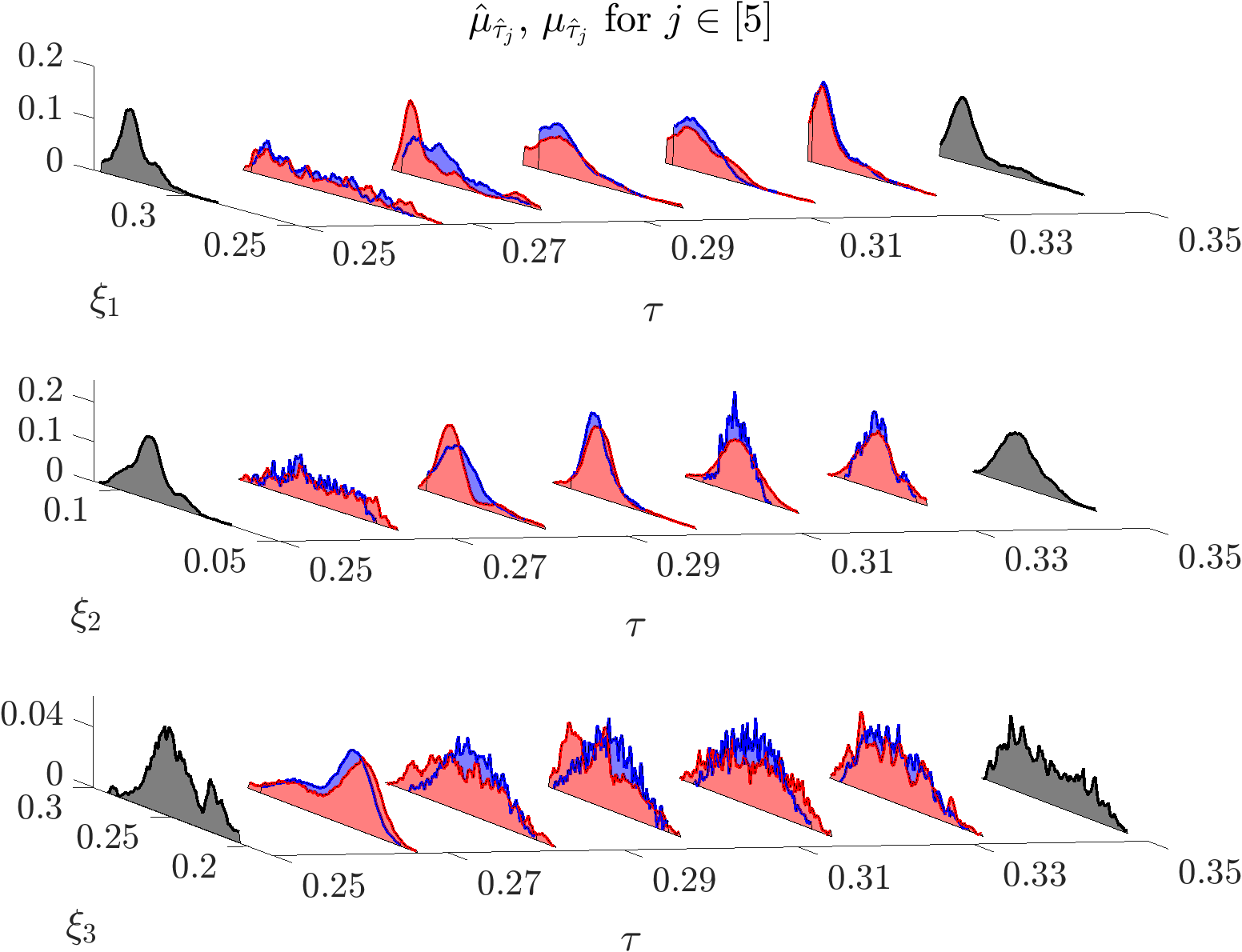}
    \caption{{\small{Predicted $\hat{\mu}_{\hat{\tau}_j}$ (\emph{blue}) vs. measured $\mu_{\hat{\tau}_j}$ (\emph{red}) at times $\hat{\tau}_{j\in\llbracket 5\rrbracket}$ during the 3rd control cycle with $s_{\rm{int}}=4$. Distributions at the control cycle boundaries are in \emph{black}.}}}
\vspace*{-0.22in}
\label{FigInterpMarginals1D}
\end{figure}

\begin{table}[h]
\centering
\begin{tabular}{| c | c | c | c | c | c |} 
 \hline
 ~$s_{\rm{int}}$~ & ~$W_{1}$~ & ~$W_{2}$~ & ~$W_{3}$~ & ~$W_{4}$~ & ~$W_{5}$~\\
 \hline\hline
 $0$ & $2.0489$ & - & - & - & - \\
 \hline
 $1$ & $2.2695$ & $1.1750$ & - & - & - \\
 \hline
 $2$ & $5.7717$ & $0.9163$ & $0.3794$ & - & - \\
 \hline
 $3$ & $2.2413$ & $1.6432$ & $1.2345$ & $0.6010$ & -  \\
\hline
 $4$ & $0.6372$ & $1.2691$ & $0.9176$ & $0.6689$ & $0.2111$ \\
\hline
\end{tabular}
\caption{{\small{Number of intracycle marginals $s_{\rm{int}}$ vs. Wasserstein distances $W_{j}$ as in \eqref{Wasserstein}. All entries are scaled up by $10^{4}$.}}}
\label{tab:WassersteinErrorsVaryingSint}
\vspace*{-0.2in}
\end{table}

\noindent \textbf{Applying the Proposed Algorithm.} Given a query context $\bm{c}$, we determine the closest CPS context for which profiling data is available, using the Euclidean distance between cyber context vectors \eqref{defcCyber}, and the Fr\'echet distance \cite{Eiter1994ComputingDF} between physical context curves \eqref{defcPhysical}. In this case study, we consider a query context with closest $\bm{c}_{\rm{cyber}}=\begin{bmatrix} 15, 15\end{bmatrix}^\top$ and closest $\bm{c}_{\rm{phys}}=y_{\rm{des}}^1(x)$. Profiling data for this $\bm{c}$ is shown in Fig. \ref{FigAllMeasuredStatesForOneContext}.

For each of the $n=500$ profiles, we are given the end times for each of the $n_{c}=5$ control cycles. We then determine the statistics of the cycle end times (Fig. \ref{FigControlCycleEndTimesUnnormalizedHistogram}) and compute the empirical distributions of $\bm{\xi}$ at the means (Table \ref{tab:ControlCycleEndTimeMeanStdDev}) of the control cycle start/end time boundaries. For empirical distributions at times between cycle boundaries, we let $s_{\rm{int}}$ be the number of marginals equispaced-in-time between each cycle boundary. We then set $\tau_{\sigma\in\llbracket s\rrbracket}$ from the means in Table \ref{tab:ControlCycleEndTimeMeanStdDev}, where $s:=1+n_c(s_{\rm{int}}+1)$ and $\tau_{\sigma(s_{\rm{int}}+1)+1}$ is the sampled mean end time for the $\sigma$th control cycle.

Our distributions are as per \eqref{EmpiricalMeasures}, where $\bm{\xi}^{i}(\tau_{\sigma})$ is the sample of the hardware resource state \eqref{defFeatureVec} at time $\approx\tau_\sigma$ (within 5ms) for profile $i$ given context $\bm{c}$. 

We set $\varepsilon=0.1$ and solve the discrete MSBP \eqref{DiscreteMSBP} with squared Euclidean cost $\bm{C}$ using \eqref{MultimargSinkFinal}. Fig. \ref{FigConvergenceHilbert} shows that the Sinkhorn iterations converge linearly (Sec. \ref{subsec:ConvergenceMultiSink}). We emphasize here that the computational complexity of proposed algorithm is minimal, thanks to the path structure of the information. In particular, we solve the MSBP \eqref{DiscreteMSBP} with $n^{s} = 500^{26}$ decision variables in approx. 10 s in MATLAB on an Ubuntu 22.04.2 LTS Linux machine with an AMD Ryzen 7 5800X CPU.

Fig. \ref{FigInterpMarginals1D} compares predicted versus observed empirical distributions. Specifically, Fig. \ref{FigInterpMarginals1D} shows $s_{\rm{int}}+1=5$ distributional predictions $\hat{\mu}_{\hat{\tau}_j}$ at times $\hat{\tau}_j$, temporally equispaced throughout the duration of the 3rd control cycle, i.e., between $\tau_{2(s_{\rm{int}}+1)+1}$ and $\tau_{3(s_{\rm{int}}+1)+1}$, with
$$ \hat{\tau}_j=\tau_{2(s_{\rm{int}}+1)+1} + \Bigg(\frac{\tau_{3(s_{\rm{int}}+1)+1}-\tau_{2(s_{\rm{int}}+1)+1}}{s_{\rm{int}}+2}\Bigg)j, $$
where $j\in\llbracket s_{\rm{int}}+1 \rrbracket$. We used \eqref{MuInterpolation} with $\sigma=2(s_{\rm{int}}+1)+j$, since $\hat{\tau}_j\in[\tau_{2(s_{\rm{int}}+1)+j},\tau_{2(s_{\rm{int}}+1)+j+1}]$.

From Fig. \ref{FigInterpMarginals1D} it is clear that the measure-valued predictions, while largely accurate, are prone to error in cases where the software resource usage behavior changes in bursts too short to be appear in our observations. It follows that increasing the number of snapshots should yield an improvement in overall accuracy. In this example, we achieve this by increasing $s_{\rm{int}}$. Table \ref{tab:WassersteinErrorsVaryingSint} reports the Wasserstein distances $W(\cdot,\cdot)$ between the corresponding predicted and measured distributions:
\begin{align}
W_j := W(\hat{\mu}_{\hat{\tau}_j},\mu_{\hat{\tau}_j}) \quad\forall j\in\llbracket s_{\rm{int}}+1 \rrbracket.
\label{Wasserstein}
\end{align}
We computed each of these $W_j$ as the square root of the optimal value of the corresponding linear program that results from specializing \eqref{DiscreteMSBP} with $s=2$, $\varepsilon=0$.


\vspace*{-0.05in}
\section{Concluding Remarks}\label{sec:conclusions}
\noindent We apply recent algorithmic advances in solving the MSBP to learn stochastic hardware resource usage by control software. The learnt model demonstrates accurate nonparametric measure-valued predictions for the joint hardware resource state at a desired time conditioned on CPS context. The formulation and its solution comes with a maximum likelihood guarantee in the space of probability measures, and the algorithm enjoys a guaranteed linear convergence rate.

\bibliographystyle{IEEEtran}
\bibliography{references.bib}

\end{document}